\documentclass{amsart}
\usepackage{amssymb,bbm,comment}
\usepackage{graphics,url,hyperref}
\usepackage{enumerate}
\usepackage{caption}
\usepackage{subcaption}
\usepackage{tikz}
\usetikzlibrary{decorations.markings}
\usepackage{pgfplots}
\usetikzlibrary{arrows,calc}

\newcommand{\ind}[1]{{\mathbf 1}_{\{#1\}}}
\def\C{\mathbb{C}}
\def\N{\mathbb{N}}
\def\E{\mathbb{E}}
\def\P{\mathbb{P}}
\def\R{\mathbb{R}}

\newcommand{\diff}{\mathop{}\mathopen{}\mathrm{d}}

\newtheorem{proposition}{Proposition}
\newtheorem{lemma}{Lemma}[proposition]
\newtheorem{theorem}{Theorem}
\newtheorem{definition}{Definition}

\def\var{\mathrm{var}}
\def\atan{\mathrm{ArcTan}}

\title[Offloading scheme for data centers]{Analysis of an offloading scheme for data centers  in the framework of fog computing}

\author[C. Fricker]{Christine Fricker}
\author[F. Guillemin]{Fabrice Guillemin}
\address[F. Guillemin]{CNC/NCA Orange Labs2, Avenue Pierre Marzin, 22300 Lannion, France}
\email{Fabrice.Guillemin@orange.com}
\author[Ph. Robert]{Philippe Robert}
\author[G. Thompson]{Guilherme Thompson${ }^\dag$}\thanks{$\dag$ G. Thompson's research was supported by Brazilian Government/CAPES grant BEX 13748-13-0}
\address[C. Fricker, Ph. Robert, G. Thompson]{INRIA Paris, 2 rue Simone Iff, CS 42112,
75589 Paris Cedex 12, France}
\email{Philippe.Robert@inria.fr}
\urladdr{http://www-rocq.inria.fr/\~{}robert}

\begin{document}

\begin{abstract}
{In the context of fog computing, we consider a simple case when data centers are installed at the edge of the network and assume that if a request arrives at an overloaded data center, then it is forwarded to a neighboring data center with some probability. Data centers are assumed to have a large number of servers and that  traffic at some of them is causing saturation. In this case the other data centers may help to cope with this saturation regime by  accepting some of the rejected requests. Our aim  is to qualitatively estimate the gain achieved via cooperation between  neighboring data centers. After proving some convergence results, related to the scaling limits of loss systems, for the process describing the number of free servers at both data centers, we show that the performance of the system can be expressed in terms of the invariant distribution of a random walk in the quarter plane. By using and developing existing results in the technical literature, explicit formulas for the blocking rates of such a system are derived. }
\end{abstract}

\maketitle

\section{Introduction}

Cloud computing has become one of the major stakes in the development of information technology by offering the possibility of reserving computing resources online. Commercial offers already exist for customers (residential or business) relying on big data centers like Amazon \nocite{AmazonEC2} or Azure \nocite{Azure} for example. This kind of technology is also relevant for network operators in the framework of network function virtualization, where network functions can be instantiated on data centers instead of dedicated hardware. In this context, there is currently a clear trend to distribute data centers. For network operators, it is possible to instantiate  at the edge of the network  functions which were so far centralized in servers (e.g., mobile core functions). Furthermore, by allocating resources closer to end users, it is expected to offer better quality of experience. Distributing cloud computing resources at the edge of the network is known as fog computing. See~\cite{Bonomi,Shenker,Wood} and~\cite{Rai}.

{Data centers involved in fog computing have  a smaller capacity than those in the case of cloud computing and therefore more subject to congestion. Hence, to reduce the probability of request blocking, fog computing data centers have to collaborate. For instance, when one request cannot be accommodated by one of them, it may be forwarded to another one.}

{A typical example of such a situation is when data centers are located on a logical ring at the edge of the network. See Figure~\ref{FogFig}. A request arriving in an overloaded data center with index $i$,  may be forwarded to a neighboring data center $i{-}1$ or $i{+}1$ with some probability.  Hence,  if the  traffic to a data center is causing saturation, the other data centers may help alleviate this saturation regime. The aim of this paper is of investigating the impact of such a cooperative scheme. In practice, the network could be backed up by a central (bigger) data center at the core of the network but at a price in terms of latency. We will not consider this additional feature here.}

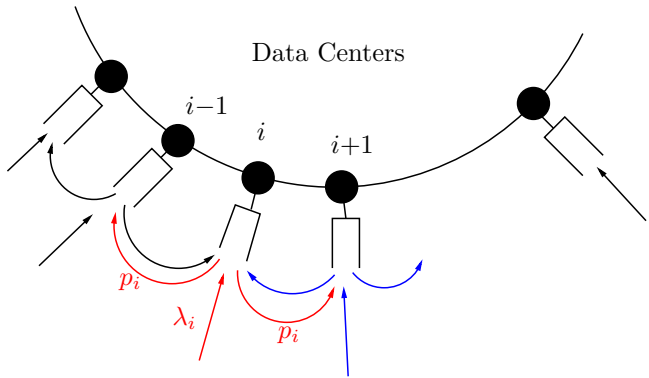
\begin{figure}[ht]
\centering 
	\scalebox{0.3}{\includegraphics{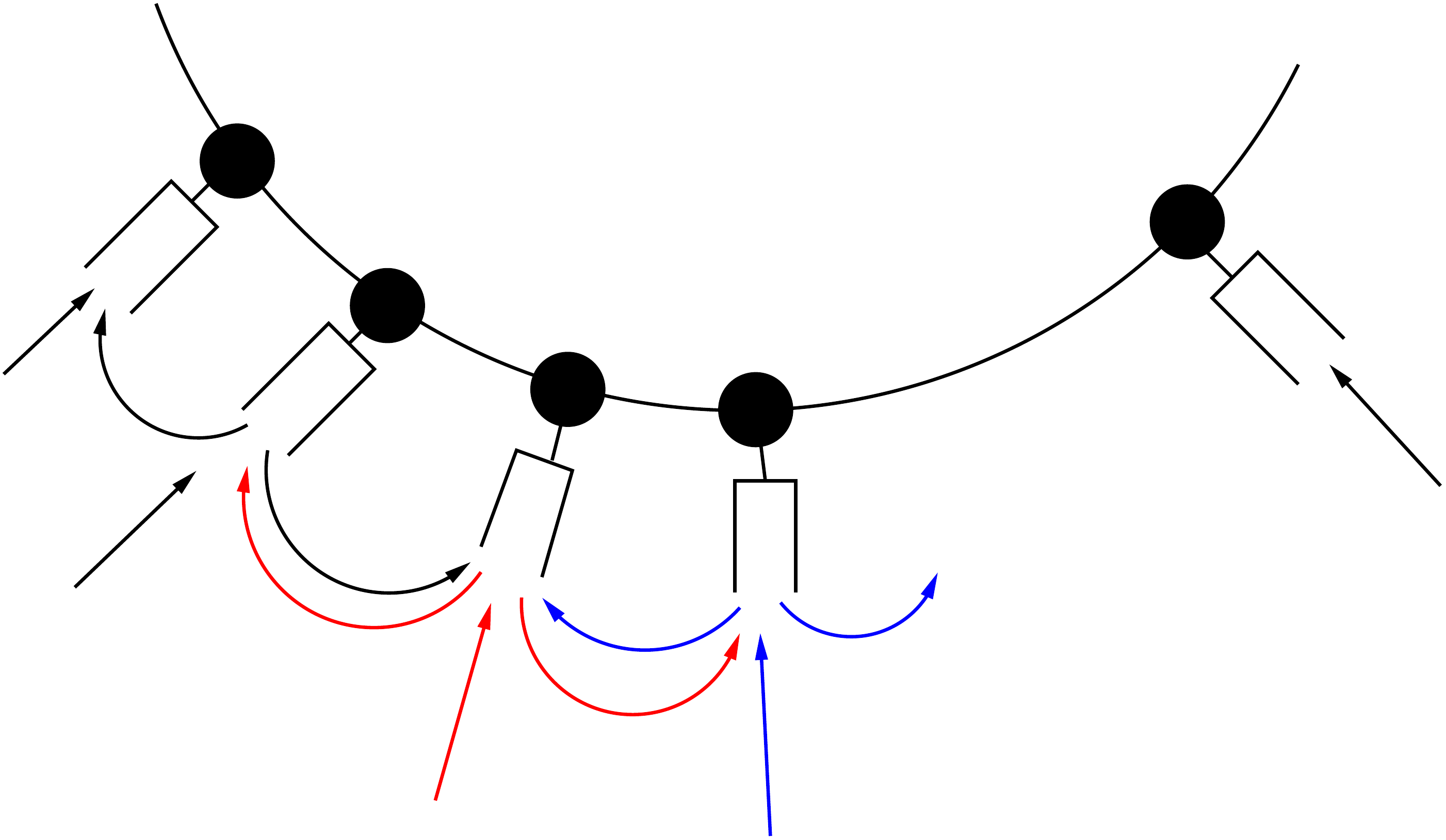}}
\put(-150,120){Data Centers}
\put(-175,100){$i{-}1$}
\put(-148,90){$i$}
\put(-120,85){$i{+}1$}
\put(-180,20){\textcolor{red}{$\lambda_i$}}
\put(-200,35){\textcolor{red}{$p_i$}}
\put(-140,15){\textcolor{red}{$p_i$}}
\caption{A Fog Computing Architecture}\label{FogFig}
\end{figure}

\subsection*{Collaboration of Two Data Centers}
Our aim here is to qualitatively estimate the gain achieved by the collaboration of data centers at the edge of the network. The main part of our analysis will concern the impact of the collaboration of two data centers. It is shown in Section~\ref{ExtSec} that the analysis applies also to more general architectures of fog computing, as in Figure~\ref{FogFig}, provided they are not  congested. 

For $i\in\{1,2\}$, the external arrival process of  requests to data center/facility \#$i$,  referred to as class $i$ requests,  is Poisson with parameter $\lambda_i$. If one of the $C_i$ servers is idle upon arrival, then the request is processed by this data center. Otherwise, if the data center is saturated, i.e., all the $C_i$ servers are busy, then with probability $p_i$ the request is forwarded to the other data center if it is not saturated too, otherwise with probability $1-p_i$ the request is rejected. A request allocated at data center \#$i$ is processed at rate $\mu_i$. 

By considering the number of requests processed at both data centers, this scheme can be clearly represented by a two dimensional Markov process on $\{0,\ldots,C_1\} {\times}\{0,\ldots,C_2\}$. This Markov process, related to loss networks, is not reversible in general and its invariant distribution {does not  have a product form expression.  Even if a numerical analysis of the equilibrium equations is always possible, it is very likely that it will not give precise qualitative and quantitative results concerning  the impact of  rerouting parameters $p_1$ and $p_2$ of the offloading scheme. Our goal is of giving {\em explicit} closed form expressions of the equilibrium probability that a request is rejected, see Theorem~\ref{TheoLoss} which is our main result in this domain. }

To overcome the difficulty of not having an explicit expression of the equilibrium, we have chosen to study a scaled version of this network. The input rates $\lambda_1$, $\lambda_2$ and the capacities $C_1$, $C_2$ are assumed to be proportional to a large parameter $N$ which goes to infinity. {This scaling has been  introduced by Kelly in the context of loss networks, see~\cite{Kelly}. As it will be seen, there is a relation  between the parameters (see Condition~(E) below), which implies that both data centers can be saturated with positive probability. We will focus mainly on this case which is, in our view, the most interesting situation to assess the benefit of offloading mechanisms in a congested environment. Otherwise, the situation is much  simpler.  One of the data centers will be underloaded, so that the rejection rate at equilibrium will  converge to $0$ as $N$ gets large, in particular external arrivals to this data center and the rerouted jobs from the other data center will be accepted with probability $1$ in the limit. See Proposition~\ref{P1} and Theorem~\ref{T1}. }

In this limiting regime we prove convergence results for the process describing the number of free servers at both data centers in the same way as in ~\cite{Hunt} for loss networks. {We show that the invariant distribution of a random walk in the quarter plane is playing a key role in the asymptotic behavior of the loss probabilities at equilibrium. The derivation of the equilibrium is based on the analysis of random walks in  $\N^2$ by~\cite{FayolleIas}.  By taking advantage of the specific characteristics of the random walks considered, we are able to get an explicit expression of the generating function of their invariant distributions in terms of elliptic integrals instead of contour integrals in the complex plane as in~\cite{FayolleIas}. See Theorem~\ref{TheoLoss}. With these results we can then  assess quantitatively the interest of this load balancing mechanism by comparing the respective loss probabilities of the two streams of requests.}

The organization of this paper is as follows: in Section~\ref{model} the stochastic model is introduced and the limit results for the scaling regime are obtained. A family of random walks is shown to play a central role. Section~\ref{LimitRW} establishes the functional relation satisfied by the generating function of the invariant measure of one of these random walks. Section~\ref{BoundValue} gives an explicit representation of this generating function in Theorem~\ref{TheoLoss} and therefore of the performance metrics of the load balancing mechanism. Section~\ref{App} presents some numerical examples of these results. Concluding remarks are presented in Section~\ref{conclusion}.

\section{Model description}\label{model}

\subsection{Model} We consider in this paper two processing facilities in parallel. The first one is equipped with $C_1$ servers and serve requests (for computing resources) arriving according to a Poisson process with rate $\lambda_1$; each request requires an exponentially distributed service time with mean $1/\mu_1$ (a request if accepted occupies a single server). Similarly, the second processing facility is equipped with $C_2$ servers and serves service requests arriving according to  a Poisson process with rate $\lambda_2$; service times are exponentially distributed with mean $1/\mu_2$.

To reduce the blocking probability, we assume that requests arriving at a service facility with no available servers are forwarded to the other one with a given probability. More precisely, if a request arrives at service facility \#1 with no available servers, the request is forwarded to the other service facility with probability $p_1$. Similarly, a request arriving at facility \#2 with no available servers is forwarded to the other facility with probability $p_2$. See Figure~\ref{fignet}.
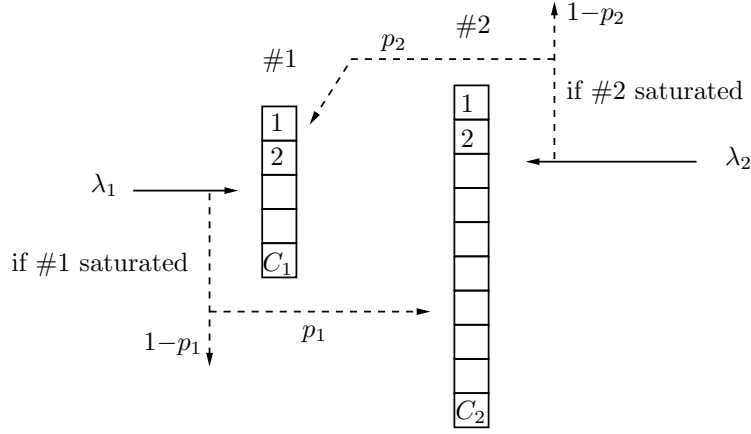
\begin{figure}[h]
\centering
	\scalebox{0.35}{\includegraphics{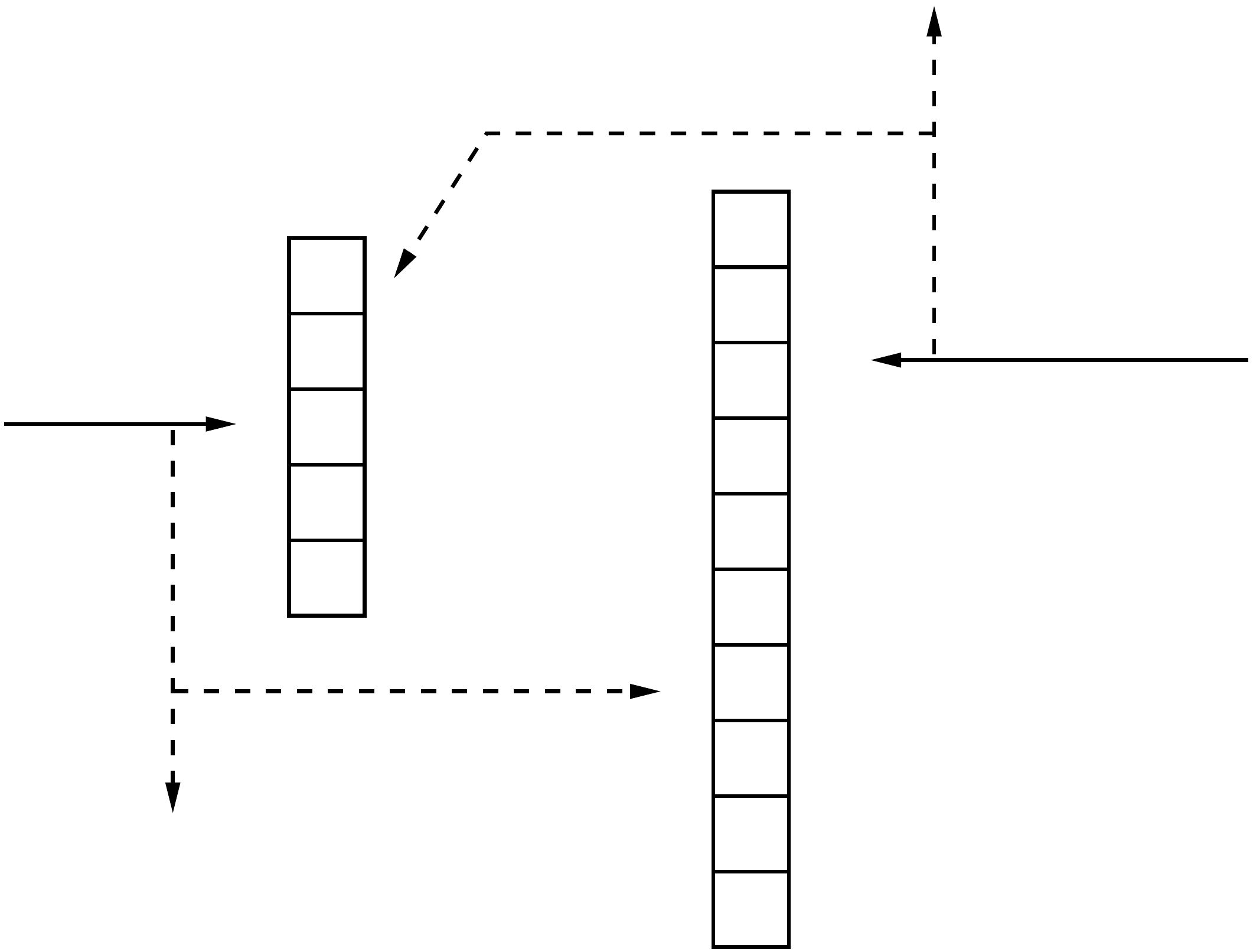}}
	\put(-165,137){\#$1$}
	\put(-165,60){$C_1$}
	\put(-230,90){$\lambda_1$}
	\put(-150,35){$p_1$}
	\put(-210,30){$1{-}p_1$}	
	\put(-90,120){$1$}
	\put(-162,113){$1$}
	\put(-260,60){if \#$1$ saturated}
	\put(-92,150){\#$2$}		
	\put(-92,4){$C_2$}
	\put(10,100){$\lambda_2$}
	\put(-120,145){$p_2$}
	\put(-50,155){$1{-}p_2$}
	\put(-90,107){$2$}
	\put(-162,100){$2$}
	\put(-50,125){if \#$2$ saturated}
	\caption{Load Balancing between Two Data Centers}\label{fignet}
\end{figure}

Let $L_1(t)$ and $L_2(t)$ denote the number of occupied servers in facilities \#1 and \#2 at time $t$, respectively. Owing to the Poisson and exponential service time assumptions, $(L(t))= ((L_1(t), L_2(t)))$ is a Markov process with values in the set $\{0, \ldots, C_1\}\times \{0,\ldots, C_2\}$, and  
 transitions from $(\ell_1,\ell_2)$ to $(\ell_1+i,\ell_2+j)$ occurring at rate
$$ \begin{cases}
	(\lambda_1+p_2 \lambda_2\ind{\ell_2=C_2})\cdot\ind{\ell_1<C_1}& \text{ if } (i,j)=(1,0),\\
	(p_1 \lambda_1\ind{\ell_1=C_1}+ \lambda_2) \cdot \ind{\ell_2<C_2} & \text{ if } (i,j)=(0,1),\\
	\mu_1\ell_1 & \text{ if } (i,j)=(-1,0),\\
	\mu_2 \ell_2& \text{ if } (i,j)=(0,-1)
\end{cases} $$
and $0$ otherwise. 

The equilibrium characteristics of this Markov process on a finite state space, like loss probabilities, do not seem to have closed form expressions in general. A scaling approach is used in the following to get some insight on the performance of such a strategy. We first introduce a random walk in $\N^2$. 

\subsection{A random walk in the extended positive quadrant}
We now consider the following random walk in the extended positive quadrant.

\begin{definition}\label{defim}
For fixed $l=(l_1,l_2)\in\R_+^2$, one defines the random walk $(\overline{m}_l(t))$ on $(\N\cup\{+\infty\})^2$ as follows:  the transition from $(m_1,m_2)$ to $ (m_1{+}a,m_2{+}b)$ occurs at rate
\begin{equation}\label{ratem}
	\begin{cases}
		\mu_1l_1 & \mbox{ if } (a,b)=(1,0),\\
		\mu_2l_2 & \mbox{ if } (a,b)=(0,1),\\
		\lambda_1+p_2\lambda_2 \ind{m_2=0} & \mbox{ if } (a,b)=(-1,0) \text{ and } m_1>0,\\
		\lambda_2 +p_1\lambda_1 \ind{m_1=0} & \mbox{ if } (a,b)=(0,-1) \text{ and } m_2>0,
	\end{cases}
\end{equation}
for $(m_1,m_2)\in(\N\cup\{+\infty\})^2$ with the convention that $+\infty\pm x=+\infty$ for $x\in\N$ (see Figure~\ref{fig_transitions_ml}).
\end{definition}

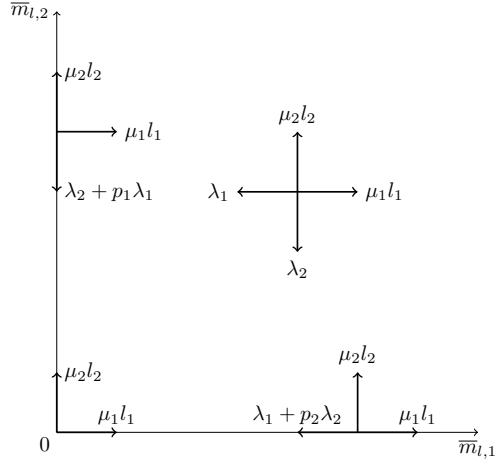
\begin{figure}
\centering
\scalebox{0.8}{\begin{tikzpicture}
\draw[->]
  (0,0) -- (7,0) node[below] {$\overline{m}_{l,1}$};
\draw[->]
  (0,0) -- (0,7) node[left] {$\overline{m}_{l,2}$};

\node at (-.2,-.2) {$0$};

\draw[->, thick]
	(0,5) -- ++	(0,1)	node[right]	{$\mu_{2} l_{2}$};
\draw[->, thick]
	(0,5) -- ++	(0,-1)	node[right]	{$\lambda_{2} + p_{1} \lambda_{1}$};  
\draw[->, thick]
	(0,5) -- ++	(1,0)	node[right]	{$\mu_{1} l_{1}$};  

\draw[->, thick]
	(5,0) -- ++	(1,0)	node[above]	{$\mu_{1} l_{1}$};
\draw[->, thick]
	(5,0) -- ++	(-1,0)	node[above]	{$\lambda_{1} + p_{2} \lambda_{2}$};  
\draw[->, thick]
	(5,0) -- ++	(0,1)	node[above]	{$\mu_{2} l_{2}$};  

\draw[->, thick]
	(0,0) -- ++	(1,0)	node[above]	{$\mu_{1} l_{1}$};
\draw[->, thick]
	(0,0) -- ++	(0,1)	node[right]	{$\mu_{2} l_{2}$};  

\draw[<->, thick]
	(3,4)	node[left]	{$\lambda_{1}$} --	(5,4)	node[right]	{$\mu_{1} l_{1}$};
\draw[<->, thick]
	(4,3)	node[below]	{$\lambda_{2}$} --	(4,5)	node[above]	{$\mu_{2} l_{2}$};
\end{tikzpicture}}
\caption{Transitions for $\bar{m}_l(t)$.}
       \label{fig_transitions_ml}
\end{figure}

In particular $(+\infty,+\infty)$ is an absorbing point for the process $(\overline{m}_l(t))$. 
The random walk $(\overline{m}_l(t))$ is a special case of the Markov process investigated in~\cite{FayolleIas}. 

The following result summarizes the stability properties of this random walk. Critical cases are omitted. 
\begin{proposition}\label{P1}
For $l=(l_1,l_2)\in\R_+^2$,
\begin{enumerate}[(i)]
\item If one of the conditions 
\begin{align*}
&a)\;  \lambda_2 < \mu_2 l_2 \text{ and }
		\lambda_1p_1 + \lambda_2 > \mu_1 l_1 p_1 + \mu_2 l_2 \\
&b) \;		 \lambda_1 < \mu_1 l_1 \text{ and }
		\lambda_1 + \lambda_2 p_2 > \mu_1 l_1 + \mu_2 l_2 p_2\\
&c) \;		\lambda_1 > \mu_1 l_1 \text{ and }\lambda_2 > \mu_2 l_2
\end{align*}
holds then the Markov process $(\overline{m}_l(t))$ is ergodic on $\N^2$. In this case the unique invariant distribution on $\N^2$ is denoted by $\pi_l$. 

\item If $\lambda_1< \mu_1l_1 \text{ and } \lambda_2 <\mu_2l_2$,
the unique invariant distribution of $(\overline{m}_l(t))$ on the extended state space $(\N\cup\{+\infty\})^2$ is the Dirac measure $\delta_{(\infty,\infty)}$.
\item If $\lambda_1> \mu_1l_1,\; \lambda_2 < \mu_2 l_2\text{ and } \lambda_1p_1+\lambda_2 < p_1\mu_1l_1+\mu_2l_2$, the unique invariant distribution of $(\overline{m}_l(t))$ on $(\N\cup\{+\infty\})^2$ is $G_{\delta_1}\otimes \delta_{\infty}$, where $G_{\delta_1}$ is the geometric distribution with parameter $\delta_1=\mu_1l_1/\lambda_1$. 
\item If $\lambda_2> \mu_2l_2,\; \lambda_1 < \mu_1 l_1 \text{ and } \lambda_2p_2+\lambda_1 < p_2\mu_2l_2+\mu_1l_1$, the unique invariant distribution of $(\overline{m}_l(t))$ on $(\N\cup\{+\infty\})^2$ is $ \delta_{\infty}\otimes G_{\delta_2}$, where $G_{\delta_2}$ is the geometric distribution with parameter $\delta_2=\mu_2l_2/\lambda_2$.
\end{enumerate}
\end{proposition}

\begin{proof}
Due to~\cite{FayolleIas}, see also~\cite[Proposition~9.15]{Robert}, $(\overline{m}_l(t))$ is ergodic if and only if one of the conditions of ($i$)  holds. {As long as it does not hit $0$, the first (resp. second)  coordinate of $(\overline{m}_l(t))$ behaves as an $M/M/1$ queue with arrival rate $\mu_1 l_1$ (resp. $\mu_2 l_2$)  and service rate $\lambda_1$ (resp.  $\lambda_2$). Under the conditions of~({\em ii}), each of these  $M/M/1$ queues is transient, in particular starting from $1$, it has a positive probability of not returning to $0$. This implies that after some random time, the process $(\overline{m}_l(t))$ stays in the interior of the quadrant $\N^2$ and therefore  behaves as a couple of independent transient $M/M/1$ queues. Consequently, both coordinates of $(m_l(t))$ are converging in distribution to $+\infty$.} Similarly, for ({\em iii}) and ({\em iv}), the process $(\overline{m}_l(t))$ can be coupled to two queues, the first one, an $M/M/1$ queue which is transient and the second one, an ergodic $M/M/1$ queue, with an invariant distribution which is geometrically distributed. 
\end{proof}

\subsection{Heavy Traffic Scaling Regime}
We investigate now the case when some of the parameters of the processing facilities are scaled up by a factor $N\in\N$. The arrival rates are given by $\lambda_1 N$ and $\lambda_2 N$ with $\lambda_1>0$ and $\lambda_2>0$. Similarly the capacities are given by $C_1^N = N c_1$ and $C_2^N= N c_2$ for some positive constants $c_1$ and $c_2$. To indicate the dependence of the numbers of idle servers upon $N$, an upper index $N$ is added to the stochastic processes. A similar approach has been used in~\cite{Hajek} to study a load balancing scheme in an Erlang system. 

We will consider the process 
\begin{equation}\label{mfree}
(m^N(t)) = (C_1^N-L_1^N(t),C_2^N-L_2^N(t))
\end{equation}
describing the number of idle servers in both processing facilities. As it will be seen, the random walks $(\overline{m}_l(t))$, $l\in\R_+^2$,  play an important role in the asymptotic behavior of $(m^N(t))$ as $N$ goes to infinity.

\begin{theorem}\label{T1}
If one of the following conditions
\begin{equation}\tag{E}
	\begin{cases}
		 \lambda_2 < \mu_2 c_2, \\
		\lambda_1p_1 + \lambda_2 > \mu_1 c_1 p_1 + \mu_2 c_2,
	\end{cases}
	\begin{cases}
		 \lambda_1 < \mu_1 c_1,\\
		\lambda_1 + \lambda_2 p_2 > \mu_1 c_1 + \mu_2 c_2 p_2,
	\end{cases}
\text{ or }
\begin{cases}
		 \lambda_1 > \mu_1 c_1,\\
		\lambda_2 > \mu_2 c_2
	\end{cases}
\end{equation}
holds, and if the initial conditions are such that $m^N(0)=m\in\N^2$ and
$$ \lim_{N\to+\infty} \left(\frac{L_1^N(0)}{N},\frac{L_2^N(0)}{N}\right)=c=(c_1,c_2) $$
then, for the convergence in distribution,
$$ \lim_{N\to+\infty} \left(\frac{L_1^N(t)}{N},\frac{L_2^N(t)}{N},\int_0^t f(m^N(u))\,\diff u\right)
=\left(c_1,c_2,t \int_{\N^2} f(x)\pi_c(\diff x)\right) $$
for any function $f$ with finite support on $\N^2$, $\pi_c$ is the invariant distribution of the process $(\overline{m}_c(t))$ defined previously.
\end{theorem}

\begin{proof}
By using the same method as in~\cite{Hunt}, one gets an analogous result to Theorem~3 of this reference. For the convergence in distribution of processes, the relation
\begin{equation}\label{hkcv}
\lim_{N\to+\infty}\left(\frac{L_1^N(t)}{N},\frac{L_1^N(t)}{N},\int_0^t\!\! f(m^N(u))\,\diff u\right)
=\left(l_1(t),l_2(t),\int_0^t\!\!\!\int_{\N^2} f(x)\pi_{l(u)}(\diff x)\diff u\right)
\end{equation}
holds, where $(l(t))=((l_1(t),l_2(t)))$ satisfying the following integral equations
$$
\begin{cases}
l_1(t)&=\displaystyle c_1+\int_0^t \left(\rule{0mm}{4mm}\lambda_1\pi_{l(u)}(\mathcal{A}_1)+p_2\lambda_2\pi_{l(u)}(\mathcal{A}_1\cap \mathcal{A}_2^c)-\mu_1l_1(u)\right)\,\diff u\\
l_2(t)&\displaystyle=c_2+\int_0^t \left(\rule{0mm}{4mm}\lambda_2\pi_{l(u)}(\mathcal{A}_2)+p_1\lambda_1\pi_{l(u)}(\mathcal{A}_2\cap \mathcal{A}_1^c)-\mu_2l_2(u)\right)\,\diff u,
\end{cases}
$$
for $i\in\{1,2\}$, $\mathcal{A}_i=\{m\in\N^2, m_i\not=0\}$ and, for $l\in\R_+^2$, $\pi_{l}$ is the {\em unique } invariant distribution of $(\overline{m}_l(t))$.

Let us assume without loss of generality  that, under condition~(E), for example the first condition of~(E) is satisfied. It will be assumed throughout the paper. It is not difficult to construct a coupling so $L_1^N(t)\geq Q_1^N(t)$ holds almost surely for all $t\geq 0$, where $(Q_1^N(t))$ is the number of jobs of an $M/M/C_1^N/C_1^N$ queue with arrival rate $\lambda_1N$ and service rate $\mu_1$. Since $\lambda_1>\mu_1 c_1$, a classical result, see Section~7 of Chapter~6 of~\cite{Robert} for example, gives the convergence in distribution $$ \lim_{N\to+\infty} \left(\frac{L_1^N(t)}{N}\right)= (c_1),$$ in particular, $(l_1(t))$ is a constant equal to $c_1$. 

If, for $i \in \{1,2\}$, $\mathcal{N}_{\lambda_iN}$ is the Poisson process of arrivals at facility $\#i$, by using the same coupling as before one gets that the number $U_2^N(t)$, arrivals at facility \#2 up to time $t$, satisfies, for $0\leq s\leq t$, $U_2^N(t)-U_2^N(s)\geq \underline{U}^N(t)- \underline{U}^N(s)$, 
$$ \underline{U}^N(t)\stackrel{\text{def.}}{=}\mathcal{N}_{\lambda_2N}[0,t]+\int_0^t \ind{Q_1^N(u-) = C_1^N,B_1(u-)=1}\mathcal{N}_{\lambda_1N}(\diff u), $$ where $(B_1(u), u\geq 0)$ is a family of independent Bernoulli random variables with parameter $p_1$. The lower bound includes the direct arrivals $\mathcal{N}_{\lambda_2N}[0,t]$ to facility $\#2$ and the rejected jobs from $\#1$. One gets that $$L_2^N(t) \geq Q_2^N(t),$$ where $(Q_2^N(t))$ is a $G/M/C_2^N/C_2^N$ queue with the arrival process $(\underline{U}^N(t))$.

The ergodic theorem gives that, almost surely
$$ \lim_{N\to+\infty} \dfrac{\underline{U}^N(t)}{N}=t\left(\lambda_2+\lambda_1p_1\left(1-\frac{\mu_1c_1}{\lambda_1}\right) \right) > \mu_2c_2 t, $$
by condition~(E). By using this relation, one can show that, for the convergence in distribution, the relation $$ \lim_{N\to+\infty} \left(\frac{Q_2^N(t)}{N}\right)=(c_2) $$ holds. In particular $(l_2(t))$ is constant and equal to $c_2$. Therefore, almost surely, $(l(t))=(c)$ holds, hence $\pi_{l(t)} = \pi_c$. Relation~\eqref{hkcv} shows that the theorem is proven.
\end{proof} 

The following proposition states that the performances of the load balancing mechanism can be expressed with the invariant distribution $\pi_c$.

\begin{proposition}
Under Condition~(E), as $N$ goes to infinity, the probability that at equilibrium a job of class $i\in\{1,2\}$ is rejected converges to $$\beta_i= {\pi}_c\left(m\in\N^2, m_i=0\right)(1-p_i)+p_i{\pi}_c\left(0,0\right),$$
where ${\pi}_c$ is the invariant distribution of $(\overline{m}_c(t))$.
\end{proposition}

\begin{proof}
Assume that $(L_1^N(t),L_2^N(t))$ is at equilibrium, the number of class~$1$ jobs rejected between $0$ and $t$ is given by
\begin{equation*}
R_1^N(t)=\int_0^t \ind{m_1^N(u-)=0,B_1(u-)=0} \mathcal{N}_{\lambda_1 N}(\diff u) +
\int_0^t \ind{m_1^N(u-)=0,m_2^N(u-)=0, B_1(u-)=1}\mathcal{N}_{\lambda_1 N}(\diff u).
\end{equation*}
The probability of rejecting a class~1 job is hence given by 
$$ \P(m_1^N(0)=0,B_1(0)=0)+\P(m_1^N(0)=0, m_2^N(0)=0,B_1(0)=1)= \frac{\E(R_1^N(1))}{\lambda_1 N}.$$

By using the martingales associated with Poisson processes, one gets
\begin{equation*}
\frac{\E(R_1^N(1))}{\lambda_1 N}=(1-p_1)\E\left(\int_0^1 \ind{m_1^N(u-)=0}\,\diff u \right)+
p_1\E\left(\int_0^1 \ind{m_1^N(u-)=0,m_2^N(u-)=0}\,\diff u \right),
\end{equation*}
one concludes with the convergence of the previous theorem. 
\end{proof}
When condition~(E) is not satisfied, one can obtain an analogous result. Its (elementary) proof is skipped. 
The results on the asymptotic blocking probability of jobs are summarized in the following proposition, where $(A)$, $(B_1)$ and $(B_2)$ are exclusive.
\begin{proposition}\label{othercases}
Let
\begin{equation*}
(A)
		\begin{cases}
		\lambda_1 {<} \mu_1 c_1\\
		\lambda_2 {<} \mu_2 c_2
	\end{cases}
	(B_1)
	\begin{cases}
		\lambda_2 {>} \mu_2 c_2\\
		\lambda_1 {+} \lambda_2 p_2 {<} \mu_1 c_1 {+} \mu_2 c_2 p_2
	\end{cases}
	(B_2)
	\begin{cases}
		\lambda_1 {>} \mu_1 c_1\\
		\lambda_2 {+} \lambda_1 p_1 {<} c_2 \mu_2 {+} \mu_1 c_1 p_1
	\end{cases}
\end{equation*}
then, at equilibrium, the loss probability of a job of class $i\in\{1,2\}$ is converging to $\beta_i$ as $N$ goes to infinity, where
\begin{equation}\label{defbetai}
	\beta_i =
	\begin{cases}
		{\pi}_c\left(m\in\N^2, m_i{=}0\right)(1{-}p_i){+}p_i{\pi}_c(0,0) & \mbox{if (E) holds},\\
		0 & \mbox{if } (A) \; \mbox{or} \; (B_i) \; \mbox{holds},\\
		(1-p_i)\left(1-\mu_ic_i/\lambda_i\right) & \mbox{if} \; (B_{3-i}) \; \mbox{holds.}
	\end{cases}
\end{equation}
\end{proposition}

{
\subsection{An Extension to Multiple Data Centers}\label{ExtSec}
In this section, it is assumed that there are $J$ data centers, for $1\leq j\leq J$, the $j$th data center has $c_jN$ servers and the external arrivals to it are Poisson with parameter $\lambda_j N$ and services are exponentially distributed with parameter $\mu_j$. If an external request at data center $j$ finds all $c_jN$ servers occupied, it is re-routed to data center $j{+}1$ (with $J{+}1=1$)  or to data center $j{-}1$ (with $0{=}J$) with probability $p_j$, otherwise it is rejected. In particular a job is rerouted with probability $2p_j$ in the case of congestion.  See Figure~\ref{FogFig}.

For  $1\leq j\leq J$, one defines the random walk $(\overline{m}_c^j(t))$ on $\N^2$ as follows: the transition from $(m,n)\in\N^2$ to $ (m{+}a,n{+}b)$ occurs at rate
\[
	\begin{cases}
		\mu_jc_j & \mbox{ if } (a,b)=(1,0),\\
		\mu_{j+1}c_{j+1} & \mbox{ if } (a,b)=(0,1),\\
		\lambda_j+p_{j+1}\lambda_{j+1} \ind{n=0} & \mbox{ if } (a,b)=(-1,0) \text{ and } m>0,\\
		\lambda_{j+1} +p_j\lambda_j \ind{m=0} & \mbox{ if } (a,b)=(0,-1) \text{ and } n>0.
	\end{cases}
\]
If one of the conditions 
\begin{equation}\tag{$E_j$}
\begin{cases}
\lambda_{j}>\mu_jc_j,\\ 
\lambda_{j+1}>\mu_{j+1}c_{j+1},
\end{cases}
\begin{cases}
\lambda_{j+1}<\mu_{j+1}c_{j+1}\\ 
\lambda_{j}{+}\lambda_{j+1}p_{j+1}{>}\mu_{j}c_{j}{+}p_{j+1}\mu_{j+1}c_{j+1},
\end{cases}
\text{ or }
\begin{cases}
\lambda_{j}<\mu_jc_j,\\ 
\lambda_{j+1}{+}\lambda_{j}p_{j}{>}\mu_{j+1}c_{j+1}{+}p_{j}\mu_{j}c_{j},
\end{cases}
\end{equation}
holds, one gets that the associated Markov process is ergodic by Proposition~\ref{P1}, one denotes by $\pi_c^j$ its invariant probability distribution. 
As before, one takes the following convention for the indices, $J{+}1{=}1$ and $1{-}1{=}J$.

We now give a version of the previous proposition in this context. 
\begin{proposition}\label{propExt}
At equilibrium, the loss probability of a job of class $j\in\{1,\ldots,J\}$ is converging to $\beta_j$ as $N$ goes to infinity in the following cases,
\begin{enumerate}
\item No Congestion.\\
If $\lambda_j<\mu_j c_j$ for all $j\in\{1,\ldots,J\}$ then $\beta_j\equiv0$.
\item One saturated node.\\
If, for some $j_0\in\{1,\ldots,J\}$,  $\lambda_{j_0}{>}\mu_{j_0} c_{j_0}$ and if $\lambda_j<\mu_j c_j$ holds for all $j\not=j_0$ and 
\[
\begin{cases}
\lambda_{j_0+1}+\lambda_{j_0}p_{j_0}<\mu_{j_0+1}c_{j_0+1}+p_{j_0}\mu_{j_0}c_{j_0}\\
\lambda_{j_0-1}+\lambda_{j_0}p_{j_0}<\mu_{j_0-1}c_{j_0-1}+p_{j_0}\mu_{j_0}c_{j_0},
\end{cases}
\]
then $\beta_j=0$ if $j\not=j_0$ and $\beta_{j_0}=(1-2p_{j_0})(1-\mu_{j_0}c_{j_0}/\lambda_{j_0})$
\item Two saturated neighboring nodes.\\
If, for some $j_0\in\{1,\ldots,J\}$,  one of the conditions~($E_{j_0}$) holds and $\lambda_j<\mu_j c_j$ holds for all $j\not=j_0, j_0{+}1$ and
\begin{equation}\label{F1}
\begin{cases}
\lambda_{j_0+2}+\lambda_{j_0+1}p_{j_0+1}\pi_c^{j_0}(\N{\times}\{0\})<\mu_{j_0+2}c_{j_0+2}\\
\lambda_{j_0-1}+\lambda_{j_0}p_{j_0}\pi_c^{j_0}(\{0\}{\times}\N)<\mu_{j_0-1}c_{j_0-1}
\end{cases}
\end{equation}
then
\[
\begin{cases}
\beta_{j_0}={\pi}^{j_0}_c\left(\{0\}{\times}\N\right)(1{-}2p_{j_0}){+}p_{j_0}{\pi}_{c}^{j_0}(0,0)\\
\beta_{j_0+1}={\pi}^{j_0}_c\left(\N{\times}\{0\}\right)(1{-}2p_{j_0+1}){+}p_{j_0+1}{\pi}_{c}^{j_0}(0,0).
\end{cases}
\]
\end{enumerate}
\end{proposition}
The proof is similar to the proof of the previous proposition and is therefore omitted. Note that Condition~\eqref{F1} implies that the nodes with index $j_0-1$ and $j_0+2$ are underloaded, so that only nodes with index $j_0$ and $j_0+1$ are congested.  This result covers  partially the set of various possibilities but, as long as only two neighboring nodes are congested, it can be extended quite easily to the case where only pairs of nodes are congested. 

When there are at least three neighboring congested nodes, this method does not apply. It occurs when one of the conditions~($E_{j_0}$) holds and one of the conditions of~\eqref{F1} is not satisfied. One has to consider the invariant distributions of a three dimensional random walk in $\N^3$ for which there are scarce results. Nevertheless this situation should be, in practice, unlikely if the fog computing architecture is conveniently designed so that a local congestion can be solved by using the neighboring resources.

This proposition shows that the evaluation of the performances of the offloading algorithm can be expressed in terms of the invariant distributions of the random walks $(\overline{m}_c(t))$ introduced in Definition~\ref{defim}. The rest of the paper is devoted to the analysis of these invariant distributions when they exist. In particular, we will derive an explicit expression of  the blocking probabilities $\beta_i$ at facility \# $i$.}

\section{Characteristics of the limiting random walk}\label{LimitRW}

\subsection{Fundamental equations}
Throughout this section, we assume that the first condition of~(E) holds. Let $m_{c,1}$ and $m_{c,2}$ denote the abscissa and the ordinate of the random walk $(\overline{m}_c(t))$ in the stationary regime. Under stability condition~(E), it is shown in \cite{FayolleIas} that the generating function of the stationary numbers $m_{c,1}$ and $m_{c,2}$, defined by $P(x,y)=\E(x^{m_{c,1}}y^{m_{c,2}})$ for complex $x$ and $y$ such that $|x|\leq 1$ and $|y|\leq 1$, satisfies the functional equation
\begin{equation}
\label{eqfund}
h_1(x,y) P(x,y) = h_2(x,y)P(x,0) + h_3(x,y)P(0,y)+h_4(x,y)\pi_c(0,0),
\end{equation}
with $\pi_c(0,0)$ standing for $P(0,0)$ and
\begin{align*}
	h_1(x,y) &= -\mu_1 c_1 x^2 y -\mu_2 c_2 x y^2 + (\lambda_1+\lambda_2+\mu_1c_1+\mu_2c_2)x y -\lambda_1 y -\lambda_2 x,\\
	h_2(x,y) &= \lambda_2 \left((1-p_2)xy-x+p_2 y\right),\\
	h_3(x,y) &= \lambda_1 \left((1-p_1)xy-y+p_1 x\right),\\
	h_4(x,y) &= (\lambda_1p_1+\lambda_2p_2)xy-p_2 \lambda_2 y - p_1 \lambda_1 x.
\end{align*}
It is worth noting that 
\begin{equation} \label{h2h3h4}
	\lambda_1 p_1 (\lambda_1+\lambda_2p_2) h_2(x,y) +\lambda_2p_2 (\lambda_2+\lambda_1p_1) h_3(x,y) - \lambda_1\lambda_2(1-p_1p_2) h_4(x,y)=0.
\end{equation}

In \cite{FayolleIas,FIM}, it is shown how to compute the unknown functions by using the zeros of the kernel $h_1(x,y)$ and the results on Riemann-Hilbert problems. In the following we briefly describe how to achieve this goal. For the system under consideration, let us recall that the performance of the system is characterized by the blocking probabilities of the two classes of customers. For customers arriving at facility \#1, the blocking probability is given by
\begin{equation}\label{defbeta1}
	\beta_1 = P(0,1)(1-p_1)+p_1\pi_c(0,0)
\end{equation}
and that for customers arriving at the second facility by
\begin{equation}\label{defbeta2}
	\beta_2 = P(1,0)(1-p_2)+p_2\pi_c(0,0).
\end{equation}

By using the normalizing condition $P(1,1)=1$, we can easily show that $$\lambda_1 + \lambda_2 p_2 P(1,0) - \mu_1 c_1 = \lambda_1 P(0,1) + \lambda_2 p_2 \pi_c(0,0)$$ and $$\lambda_2 + \lambda_1 p_1 P(0,1) - \mu_2 c_2 = \lambda_2 P(1,0) + \lambda_1 p_1 \pi_c(0,0).$$

We then deduce that 
\begin{align}\label{P01}
P(0,1) = \frac{\lambda_1-\mu_1c_1+p_2(\lambda_2-\mu_2c_2)-p_2(\lambda_2+\lambda_1p_1)\pi_c(0,0)}{(1-p_1 p_2)\lambda_1}
\intertext{and}
P(1,0) = \frac{\lambda_2-\mu_2c_2+p_1(\lambda_1-\mu_1c_1)-p_1(\lambda_1+\lambda_2p_2)\pi_c(0,0)}{(1-p_1p_2)\lambda_2}.\label{P10}
\end{align}
The above relations show that the blocking probabilities $\beta_1$ and $\beta_2$ can be estimated as soon as the quantity $\pi_c(0,0)$ is known.

\subsection{Zero pairs of the kernel}

The kernel $h_1(x,y)$ has already been studied in \cite{FayolleIas} in the framework of coupled servers. For fixed $y$, the kernel $h_1(x,y)$ has two roots $X_0(y)$ and $X_1(y)$. By using the usual definition of the square root such that $\sqrt{a}>0$ for $a>0$, the solution which is null at the origin and denoted by $X_0(y)$, is defined and analytic in $\C\setminus ([y_1,y_2]\cup [y_3,y_4])$ where the reals $y_1, y_2, y_3$ and $y_4$ are such that $0<y_1<y_2<1<y_3<y_4$. The other solution $X_1(y)$ is meromorphic in $\C\setminus ([y_1,y_2]\cup [y_3,y_4])$ with a pole at 0. The function $X_0(y)$ is precisely defined by
$$ X_0(y) = \frac{-(\mu_2c_2y^2-(\lambda_1+\lambda_2+\mu_1c_1+\mu_2c_2)y+\lambda_2) +\sigma_1(y)}{2\mu_1c_1y}$$ with $$ \Delta_1(y) = (\mu_2 c_2 y^2 - (\lambda_1 + \lambda_2 + \mu_1 c_1 + \mu_2 c_2) y + \lambda_2)^2 - 4\mu_1 c_1 \lambda_1 y^2,$$
where $\sigma_1(y)$ is the analytic continuation in $\C\setminus ([y_1,y_2]\cup [y_3,y_4])$ of the function $\sqrt{\Delta_1(y)}$ defined in the neighborhood of 0. The other solution $X_1(y) ={\lambda_1}/{(\mu_1c_1X_0(y))}$.

When $y$ crosses the segment $[y_1,y_2]$, $X_0(y)$ and $X_1(y)$ describe the circle $C_{r_1}$ with center 0 and radius $r_1=\sqrt{{\lambda_1}/{(\mu_1c_1)}} > 1$, since from the first of condition of~(E), we have $\lambda_1 > \mu_1 c_1$.

Similarly, for fixed $x$, the kernel $K(x,y)$ has two roots $Y_0(x)$ and $Y_1(x)$. The root $Y_0(x)$, which is null at the origin, is analytic in $\C\setminus ([x_1,x_2]\cup [x_3,x_4])$ where the reals $x_1$, $x_2$, $x_3$ and $x_4$ are such that with $0<x_1<x_2<1<x_3<x_4$ and is given by
$$
Y_0(x) = \frac{-(\mu_1c_1x^2-(\lambda_1+\lambda_2+\mu_1c_1+\mu_2c_2)x+\lambda_1) +\sigma_2(x)}{2\mu_2c_2x}
$$
with
$$
\Delta_2(x) = (\mu_1c_1x^2-(\lambda_1+\lambda_2+\mu_1c_1+\mu_2c_2)x+\lambda_1)^2-4\mu_2c_2\lambda_2 x^2,
$$
where $\sigma_2(x)$ is the analytic continuation in  $\C\setminus ([x_1,x_2]\cup [x_3,x_4])$ of the function $\sqrt{\Delta_2(x)}$ defined in the neighborhood of 0.  The other root $Y_1(x) = {\lambda_2}/{(\mu_2c_2 Y_0(x))}$ and is meromorphic in $\C\setminus ([x_1,x_2]\cup [x_3,x_4])$ with a pole at the origin. 

When $x$ crosses the segment $[x_1,x_2]$, $Y_0(y)$ and $Y_1(y)$ describe the circle $C_{r_2}$ with center 0 and radius $r_2 = \sqrt{{\lambda_2}/{(\mu_2c_2)}}$.

\section{Boundary value problems}\label{BoundValue}

\subsection{Problem formulation}

In \cite{FIM}, it is proven that the functions $P(x,0)$ and $P(0,y)$ can be extended as meromorphic functions in $\C\setminus[x_3,x_4]$ and $\C\setminus[y_3,y_4]$, respectively. By using the fact that $X_0(y)$ and $X_1(y)$ are on circle $C_{r_1}$ for $y \in [y_1,y_2]$, we easily deduce that the function $P(x,0)$, analytic in $D_{r_1}$ (the disk with center 0 and radius $r_1$), is such that for $x \in C_{r_1}$ 
\begin{equation}
\label{RHPx}
\Re\left( i\frac{h_2(x,Y_0(x))}{h_3(x,Y_0(x))}P(x,0)\right) = \Im\left(\frac{h_4(x,Y_0(x))}{h_3(x,Y_0(x))}\pi_c(0,0) \right)
\end{equation}
where $Y_0(x) \in [y_1,y_2]$.

Similarly, the function $P(0,y)$ is analytic in $D_{r_2}$, which is the disk with center 0 and radius $r_2$, and for $x \in C_{r_2}$, we have
\begin{equation}
\label{RHPy}
\Re\left( i\frac{h_3(X_0(y),y)}{h_2(X_0(y),y)}P(0,y)\right) = \Im\left(\frac{h_4(X_0(y),y)}{h_2(X_0(y),y)}\pi_c(0,0) \right).
\end{equation}

By using Equation~\eqref{h2h3h4}, we have
$$
\Im\left(\frac{h_4(x,y)}{h_2(x,y)}\right)= -\frac{p_2(\lambda_2+\lambda_1p_1)}{\lambda_1(1-p_1p_2)} \Re\left(i\frac{h_3(x,y)}{h_2(x,y)}\right).
$$
Equation~\eqref{RHPy} can then be rewritten as
\begin{equation}
\label{RHPyb}
\Re\left( i\frac{h_3(X_0(y),y)}{h_2(X_0(y),y)}\left(P(0,y)+\frac{p_2(\lambda_2+\lambda_1p_1)}{\lambda_1(1-p_1p_2)}\pi_c(0,0)\right)\right) =0.
\end{equation}

Similarly, Equation~\eqref{RHPx} can be rewritten as
\begin{equation}
\label{RHPxb}
\Re\left( i\frac{h_2(x,Y_0(x))}{h_3(x,Y_0(x))}\left(P(x,0)+\frac{p_1(\lambda_1+\lambda_2p_2)}{\lambda_2(1-p_1p_2)}\pi_c(0,0)\right)\right) =0.
\end{equation}

Problem~\eqref{RHPxb} corresponds to Problem~(7.6) in \cite{FayolleIas} for which $i_1=i_2=i_3=0$ in the notation of that paper. The ratio ${h_2(x,Y_0(x))}/{h_3(x,Y_0(x))}$ corresponds to the function $J(x)$ in \cite{FayolleIas}.

In the following, we focus on Riemann-Hilbert problem~\eqref{RHPyb}. The analysis of problem~\eqref{RHPxb} is completely symmetrical. Moreover, to compute the blocking probabilities $\beta_1$ and $\beta_2$, we only need to compute the quantity $\pi_c(0,0)$.

\subsection{Problem resolution}

The function $P(0,y)$ is analytic in the open disk $D_{r_2}$. By using the reflection principle \cite{Cartan}, the function 
$$
y \mapsto \overline{P\left(0, {r^2_2}/{\overline{y}} \right)}
$$
is analytic on the outside of the closed disk $\overline{D_{r_2}}$. It is then easily checked that if we define 
\begin{equation}
\label{FYP0y}
F_Y(y) = 
\begin{cases}
\displaystyle P(0,y) +\frac{p_2(\lambda_2+\lambda_1p_1)}{\lambda_1(1-p_1p_2)}\pi_c(0,0), & y \in D_{r_2},\\ \\
\displaystyle \overline{P\left(0, {r_2^2}/{\overline{y}} \right)} +\frac{p_2(\lambda_2+\lambda_1p_1)}{\lambda_1(1-p_1p_2)}\pi_c(0,0), & y \in \C\setminus \overline{D_{r_2}},
\end{cases}
\end{equation}
the function $F_Y(y)$ is sectionally analytic with respect to the circle $C_{r_2}$,  the quantity  $F_Y(y)$ tends to $\pi_c(0,0){(\lambda_1+\lambda_2p_2)}/{(\lambda_1(1-p_1p_2))}$ when $y$ goes to infinity, and for $y \in C_{r_2}$
\begin{equation}
\label{Hilberty}
F^i_Y(y) = \alpha_Y(y)F^e_Y(y),
\end{equation}
where $F^i_Y(y)$ (resp. $F^e_Y(y)$) is the interior (resp. exterior) limit of the function $F_Y(y)$ at the circle $C_{r_2}$, and the function $\alpha_Y(y)$ is defined on $C_{r_2}$ by
\begin{equation}
\label{defalphayini}
\alpha_Y (y) = \frac{\overline{a_Y(y)}}{a_Y(y)}
\end{equation}
with
$$
a_Y(y) = \frac{h_3(X_0(y),y)}{h_2(X_0(y),y)}.
$$

The solutions to Riemann-Hilbert problems of form~\eqref{Hilberty} are given in \cite{Lions}. We first have to determine the index of the problem defined as
$$
\kappa_Y = \frac{1}{2\pi}\var_{y \in C{r_2}} \arg \alpha_Y(y) .
$$
In \cite[Theorem~7.2]{FayolleIas} it is shown that the stability condition~(E) is equivalent to $\kappa_Y=0$.

To obtain explicit expressions, let us first study the function $\alpha_Y(y)$, which can be expressed as follows. 

\begin{lemma}
The function $\alpha_Y(y)$ defined for $y \in C_{r_2}$ by Equation~\eqref{defalphayini} can be extended as a meromorphic function in $\C\setminus([y_1,y_2]\cup[y_3,y_4])$ by setting
\begin{equation}
\label{defalphaY}
\alpha_Y(y) = \frac{\lambda_2(1-p_1p_2)X_0(y) + yR_Y(X_0(y))}{ y(\mu_2c_2(1-p_1p_2)y X_0(y) + R_Y(X_0(y))) },
\end{equation}
where
\begin{equation}
\label{defRY}
R_Y(x) = p_1\mu_1c_1(1-p_2)x^2+(p_1p_2(\mu_1c_1+\mu_2c_2)-(1-p_2)(\lambda_2+\lambda_1p_1))x -p_2(\lambda_2+\lambda_1p_1).
\end{equation}
\end{lemma}

\begin{proof}
We have for $(x,y)$ such that $y\in C_{r_2}$ and $x=X_0(y)$
\begin{multline*}
h_3(x,\bar{y})h_2(x,y) = \lambda_1\lambda_2 \left(((1-p_1)x-1)((1-p_2)x+p_2) y \bar{y} \right. \\ \left. -((1-p_1)x-1)x\bar{y}+p_1x((1-p_2)x+p_2)y -p_1x^2\right)
\end{multline*}
By using the fact that $y\bar{y}={\lambda_2}/{(\mu_2c_2)}$ and $h_1(x,y)=0$, we deduce that 
$$
h_3(x,\bar{y})h_2(x,y) = -\frac{\lambda_1\lambda_2(x-1)}{\mu_2c_2 y}\left(\lambda_2(1-p_1p_2)x+yR_Y(x)\right),
$$
where $R_Y(x)$ is defined by Equation~\eqref{defRY}, and the result follows.
\end{proof}

Since the index of the Riemann-Hilbert~\eqref{Hilberty} is null, the solution is as follows.

\begin{lemma}
\label{lemHilbert}
The solution to the Riemann-Hilbert problem~\eqref{Hilberty} exists and is unique and given for $y \in D_{r_2}$ by
\begin{equation}
\label{defFy}
F_Y(y) = \frac{\lambda_1+ \lambda_2p_2}{\lambda_1(1-p_1p_2)}\pi_c(0,0) \varphi_Y(y),
\end{equation}
where
\begin{equation}
\label{defvarphiy}
\varphi_Y(y) = \exp\left( \frac{y}{\pi} \int_{x_1}^{x_2}\frac{(\mu_1c_1 x^2-\lambda_1)\Theta_Y(x)}{xh_1(x,y)}\,\diff x \right)
\end{equation}
and
\begin{multline}
\label{defThetaY}
\Theta_Y(x) = \\ \atan\left(\frac{(1-p_1p_2)\sqrt{-\Delta_2(x)}}{(1-p_1p_2)(\mu_1c_1x^2-(\lambda_1+\lambda_2+\mu_1c_1+\mu_2c_2)x+\lambda_1)-2R_Y(x)} \right) .
\end{multline}
\end{lemma}

\begin{proof}
Since the index of the Riemann-Hilbert~\eqref{Hilberty} is null, the solution reads \cite{Lions}
$$
F_Y(y) = \phi_Y(y) \exp\left(\frac{1}{2i\pi}\int_{C_{r_2}} \frac{\log\alpha_Y(z)}{z-y} \,\diff z \right)
$$
where the function $\alpha_Y(y)$ is defined by Equation~\eqref{defalphaY} and $\phi_Y(y)$ is a polynomial. Since we know that $F_Y(y)\to\pi_c(0,0){(\lambda_1+ \lambda_2p_2)}/{(\lambda_1(1-p_1p_2))}$ as $|y|\to \infty$, then
$$
\phi_Y(y) = \frac{\lambda_1+ \lambda_2p_2}{\lambda_1(1-p_1p_2)}\pi_c(0,0).
$$

Let for $y \in C_{r_2}$ and $y = Y_0(x+i0)$ for $x\in [x_1,x_2]$
$$
\Theta_Y(x) = \arg\left(-\mu_2c_2(1-p_1p_2)Y_0(x+0i) x - R_Y(x) \right)
$$
By using the expression of $Y_0(x)$, Equation~\eqref{defThetaY} follows. It is clear that 
$$
\log\alpha_Y(Y_0(x+0i)) = -2i\Theta_Y(x).
$$

Since $Y_0(x+0i) = \overline{Y_0(x-0i)}$, we have
\begin{align*}
\frac{1}{2i\pi}\int_{C_{r_2}} \frac{\log\alpha_Y(z)}{z-y} \,\diff z &=\frac{1}{\pi} \int_{x_1}^{x_2}\Im\left(\frac{\log\alpha_Y(Y_0(x+0i))}{Y_0(x+0i)-y}\frac{\diff Y_0}{\diff x}(x+0i) \right)\,\diff x \\
&=\frac{1}{\pi} \int_{x_1}^{x_2}\Im\left(\frac{-2i}{Y_0(x+0i)-y}\frac{\diff Y_0}{\diff x}(x+0i) \right)\Theta_Y(x)\,\diff x.
\end{align*}
It is easily checked from the equation $h_1(x,Y_0(x) )=0$ that
$$
\frac{\diff Y_0}{\diff x} = \frac{-2\mu_1c_1xY_0(x)-\mu_2c_2 Y_0(x)^2+(\lambda_1+\lambda_2+\mu_1c_1+\mu_2c_2)Y_0(x) -\lambda_2}{\mu_1c_1x^2+2\mu_2c_2 xY_0(x)-(\lambda_1+\lambda_2+\mu_1c_1+\mu_2c_2)x+\lambda_1}.
$$
For $x \in [x_1,x_2]$, we have
$$
 \mu_1c_1x^2+2\mu_2c_2 xY_0(x+0i)-(\lambda_1+\lambda_2+\mu_1c_1+\mu_2c_2)x+\lambda_1 =-i\sqrt{-\Delta_2 (x)}
$$
By using once again $h_1(x,Y_0(x+0i))=0$, we obtain for $x \in [x_1,x_2]$
$$
\frac{\diff Y_0}{\diff x}(x+0i) =\frac{(\lambda_1 -\mu_1c_1x^2)Y_0(x+0i)}{-ix\sqrt{-\Delta_2(x)}}
$$
and then for real $y$
$$
\Im\left(\frac{-2i}{Y_0(x+0i)-y}\frac{\diff Y_0}{\diff x}(x+0i) \right) = \frac{(\mu_1c_1x^2-\lambda_1)y}{xh_1(x,y)}.
$$
It follows that for real $y$
$$
\frac{1}{2i\pi}\int_{C_{r_2}} \frac{\log\alpha_Y(z)}{z-y} \,\diff z = \frac{y}{\pi} \int_{x_1}^{x_2}\frac{(\mu_1c_1x^2-\lambda_1 )\Theta_Y(x)}{xh_1(x,y)}\,\diff x
$$
It is easily checked that the function on the right hand side of the above equation can analytically be continued in the disk $D_{r_2}$ and the result follows. 
\end{proof}

In view of the above lemma, we can state the main result of this section.

\begin{theorem}
The function $P(0,y)$ can be defined as a meromorphic function in $\C\setminus[y_3,y_4]$ by setting 
\begin{equation}
\label{defP0y}
 {P(0{,}y)} {=}
\begin{cases}
\displaystyle \frac{\lambda_1+ \lambda_2p_2}{\lambda_1(1-p_1p_2)}\pi_c(0,0) \varphi_Y(y) - \frac{p_2(\lambda_2+\lambda_1p_1)}{\lambda_1(1-p_1p_2)} \pi_c(0,0), \; y \in D_{r_2}, \\ \\ 
\displaystyle\frac{\lambda_1{+} \lambda_2p_2}{\lambda_1(1{-}p_1p_2)}\pi_c(0{,}0) \alpha_Y(y)\varphi_Y(y) {-} \frac{p_2(\lambda_2{+}\lambda_1p_1)}{\lambda_1(1{-}p_1p_2)} \pi_c(0{,}0), \;y \in\C\setminus\overline{D_{r_2}},
\end{cases}
\end{equation}
where $\varphi_Y(y)$ is defined by Equation~\eqref{defvarphiy}.
\end{theorem}

\begin{proof}Since the solution to the Riemann-Hilbert problem~\eqref{RHPyb} is unique, the function $P(0,y)$ coincides with the function $$F_Y(y)+ \frac{p_2(\lambda_2+\lambda_1p_1)}{\lambda_1(1-p_1p_2)}\pi_c(0,0)$$ in $D_{r_2}$. We can extend this function as follows. Noting that the function $\log\alpha_Y(y)$ is analytic in a neighborhood of the circle $C_{r_2}$, the function 
$$
y \mapsto \exp\left( \frac{1}{2i\pi}\int_{C_{r_2}} \frac{\log\alpha_Y(z)}{z-y} \,\diff z \right)
$$
defined for $y \in D_{r_2}$ can be continued as a meromorphic function in $\C\setminus[x_3,x_4]$ by considering the function defined for $y \in \C\setminus\overline{D_{r_2}},$ by
\begin{equation*}
\alpha_Y(y) \exp\left(\frac{1}{2i\pi}\int_{C_{r_2}} \frac{\log\alpha_Y(z)}{z-y} \,\diff z \right) =  \alpha_Y(y) \exp\left( \frac{y}{\pi} \int_{x_1}^{x_2}\frac{(\mu_1 c_1 x^2-\lambda_1)\Theta_Y(x)}{xh_1(x,y)}\,\diff x \right),
\end{equation*}
where the last equality is obtained by using the same arguments as above (consider first real $y$ and then extend the function by analytic continuation).
\end{proof}

For the system under consideration, let us recall that the performance of the system is characterized by the blocking probabilities of the two classes of customers. The following theorem summarizes the main results of the paper for Condition~(E). Proposition~\ref{othercases} covers the other cases. 
\begin{theorem}\label{TheoLoss}
Under Condition~(E), as $N$ goes to infinity, the probability that at equilibrium a job of facility \#$i$, $i\in\{1,2\}$ is rejected converges to $\beta_i$ with 
\begin{align*}
\displaystyle\beta_1 = \frac{(\lambda_1{-}\mu_1c_1{+}p_{2}(\lambda_{2}{-}\mu_{2}c_{2}){-}p_{2}(\lambda_{2}{+}\lambda_1p_1)\pi_c(0,0))(1{-}p_1)}{\lambda_1(1{-}p_1p_{2})}{+}p_1\pi_c(0,0) \\
\displaystyle\beta_2= \frac{(\lambda_2{-}\mu_2c_2{+}p_{1}(\lambda_{1}{-}\mu_{1}c_{1}){-}p_{1}(\lambda_{1}{+}\lambda_2p_2)\pi_c(0,0))(1{-}p_2)}{\lambda_2(1{-}p_1p_{2})}{+}p_2\pi_c(0,0)
\end{align*}
and the quantity $\pi_c(0,0)$ is given by
\begin{equation} \label{P00}
\pi_c(0,0) = 
\begin{cases}
\displaystyle\frac{ \lambda_1+\lambda_2p_2-(\mu_1c_1+\mu_2c_2p_2)}{\displaystyle (\lambda_1+\lambda_2p_2)\varphi_Y(1)} & \mbox{if } \lambda_2> \mu_2c_2,\\ \\
\displaystyle\frac{ \lambda_2+\lambda_1p_1-\mu_1c_1p_1-\mu_2c_2}{\displaystyle p_1 (\lambda_1+\lambda_2p_2)\varphi_Y(1)} & \mbox{if } \lambda_2< \mu_2c_2,
\end{cases}
\end{equation}
where $\varphi_Y(y)$ is defined by Equation~\eqref{defvarphiy}.
\end{theorem}
\begin{proof}
In the case $\lambda_2> \mu_2c_2$, the result easily follows by using Equation~\eqref{defP0y} for $y=1$ and Equation~\eqref{P01}.

In the case $\lambda_2 < \mu_2c_2$ (and then $\lambda_1>\mu_1c_1$ by Condition~(E)), we have $X_0(1) =1 $ and then, by Relation~\eqref{defalphaY}, one gets the expression for $\alpha_Y(1)$,
$$\alpha_Y(1) = p_1\frac{\lambda_1+\lambda_2p_2-\mu_1c_1-\mu_2c_2p_2}{\lambda_2+\lambda_1p_1-\mu_1c_1p_1-\mu_2c_2}.$$
Equation~\eqref{P00} then easily follows. The formulas for the blocking probabilities are obtained by using Relations~\eqref{defbeta1} and~\eqref{defbeta2} for $\beta_1$ and $\beta_2$ and the expressions~\eqref{P01} and~\eqref{P10} for $P(0,1)$ and $P(1,0)$. 
\end{proof}

To conclude this section, it is worth noting that the computation of the function $\varphi_Y(y)$ in the quantity $\pi_c(0,0)$ involves elliptic integrals. In addition, a similar result holds for the function $P(x,0)$. This enables us to completely compute the generating function $P(x,y)$.

\section{Numerical results: Offloading small data centers}\label{App}

In this section, we illustrate the results obtained in the previous sections (in particular Theorem~\ref{TheoLoss}) in order to estimate the gain achieved by the offloading scheme. We assume that the service rate at  both facilities is the same and taken equal to unity ($\mu_1=\mu_2=1$). Assume in addition that the first data center has a capacity much smaller than the second one, e.g., $c_1=1$ and $c_2=10$.

We consider the case when all the requests blocked at the first data center are forwarded to the second one ($p_1=1$) and none blocked at the second data center is forwarded to the first one ($p_2=0$). 

In Figures~\ref{ffig1} and~\ref{ffig2}, when the arrival rate $\lambda_1$ at the first data center increases, the loss rate $\beta_1$ goes from $0$ if $(A)$ or $(B_1)$ holds to a positive value if $(B_2)$ or  $(E)$ holds. For example in Figure~\ref{ffig1}, for $p_1=1$, we can see the transition from $(B_2)$ to $(E)$ when $\lambda_1=3$, and for $p_1=0.7$, the transition from $(B_2)$ to $(E)$ when $\lambda_1=1+2/0.7\simeq 3.85$. We can checked that $\beta_1$ is a continuous and not differentiable function of $\lambda_1$ at $1+2/0.7$. If $p_1=0.35$ or $p_1=0$, $(E)$ holds for the range of values $[1,5]$ considered here for $\lambda_1$.
In Figure~\ref{ffig1}, $\lambda_2=12$ thus $(E)$ holds for $\lambda_1\in [1,5]$, as $\lambda_1> \mu_1 c_1$ and $\lambda_2> \mu_2 c_2$.

In conclusion, Figures~\ref{ffig1} and~\ref{ffig2} show that the offloading mechanism improves a lot the loss rate $\beta_1$ of the requests of class $1$ and does not significantly deteriorate the corresponding performances at facility \#$2$ in the case of systematic rerouting ($p_1=1$), even when this data center is already significantly loaded as in Figure~\ref{ffig2}~(B). This means that offloading small date centers with a big back-up data center is a good strategy to reduce blocking in fog computing.

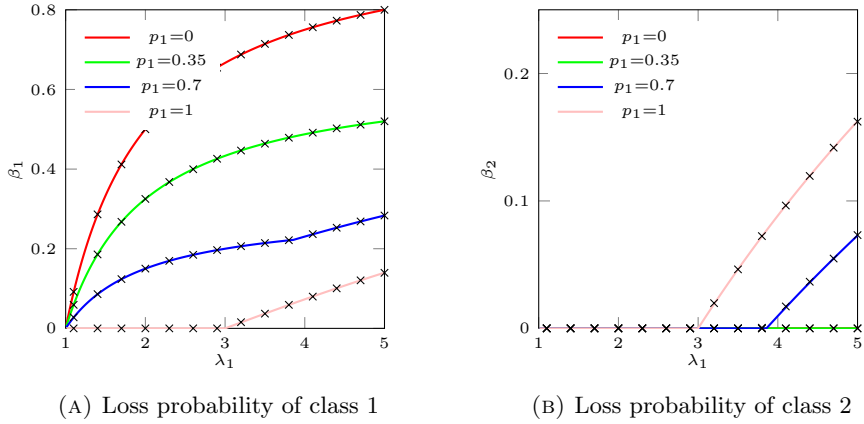
\begin{figure} [htbp]
        \begin{subfigure}[b]{0.46\textwidth}
                \begin{tikzpicture}
                 \begin{axis}[legend style={draw=none},
                                    font=\tiny,
                                       label style={font=\tiny},
      			           xlabel={$\lambda_1$},
			           xlabel style={yshift=+9pt},		
                                       ylabel= {$\beta_1$},
                                       ylabel style={yshift=-16pt},
                                       xmin=1, xmax=5, ymin=0.0, ymax=0.8,
                                      width=\linewidth,height=\linewidth,legend pos= north west]
                \addplot[thick, smooth,red] table[x=lbd1, y=beta1_p1_0.00] {ffig1.dat};
\addlegendentry{$p_1{=}0$}
\addplot[thick, smooth,green] table[x=lbd1, y=beta1_p1_0.35] {ffig1.dat};
\addlegendentry{$p_1{=}0.35$}
\addplot[thick, smooth,blue] table[x=lbd1, y=beta1_p1_0.70] {ffig1.dat};
\addlegendentry{$p_1{=}0.7$}
\addplot[thick, smooth,pink] table[x=lbd1, y=beta1_p1_1.00] {ffig1.dat};
\addlegendentry{$p_1{=}1$}
\addplot[only marks,black,mark=x] table[x=lbd1, y=simu_beta1_p1_0.00] {ffig1_simu.dat};
\addplot[only marks,black,mark=x] table[x=lbd1, y=simu_beta1_p1_0.35] {ffig1_simu.dat};
\addplot[only marks,black,mark=x] table[x=lbd1, y=simu_beta1_p1_0.70] {ffig1_simu.dat};
\addplot[only marks,black,mark=x] table[x=lbd1, y=simu_beta1_p1_1.00] {ffig1_simu.dat};
                \end{axis}
      \end{tikzpicture}
       \hspace{0.05\linewidth}
       \caption{ Loss probability of class $1$}
       \label{ffig1_a}
        \end{subfigure}
   \quad
  \begin{subfigure}[b]{0.46\textwidth}
                \begin{tikzpicture}
                 \begin{axis}[legend style={draw=none},
                                     font=\tiny,
                                       label style={font=\tiny},
      			           xlabel={$\lambda_1$},
			           xlabel style={yshift=+9pt},		
                                       ylabel= {$\beta_2$},
                                       ylabel style={yshift=-16pt},
                                       xmin=1, xmax=5, ymin=0., ymax=0.25,
                                      width=\linewidth,height=\linewidth,legend pos=  north west]
                \addplot[thick, smooth,red] table[x=lbd1, y=beta2_p1_0.00] {ffig1.dat};\addlegendentry{$p_1{=}0$}
\addplot[thick, smooth,green] table[x=lbd1, y=beta2_p1_0.35] {ffig1.dat};
\addlegendentry{$p_1{=}0.35$}
\addplot[thick, smooth,blue] table[x=lbd1, y=beta2_p1_0.70] {ffig1.dat};
\addlegendentry{$p_1{=}0.7$}
\addplot[thick, smooth,pink] table[x=lbd1, y=beta2_p1_1.00] {ffig1.dat};
\addlegendentry{$p_1{=}1$}
\addplot[only marks,black,mark=x] table[x=lbd1, y=simu_beta2_p1_0.00] {ffig1_simu.dat};
\addplot[only marks,black,mark=x] table[x=lbd1, y=simu_beta2_p1_0.35] {ffig1_simu.dat};
\addplot[only marks,black,mark=x] table[x=lbd1, y=simu_beta2_p1_0.70] {ffig1_simu.dat};
\addplot[only marks,black,mark=x] table[x=lbd1, y=simu_beta2_p1_1.00] {ffig1_simu.dat};
                \end{axis}
      \end{tikzpicture}
       \hspace{0.05\linewidth}
       \caption{Loss probability of class $2$}\label{ffig1_b}
       \end{subfigure}
\caption{Loss probabilities as a function of $\lambda_1$ with $\lambda_2{=}8$, $c_1{=}1$, $c_2{=}10$,  $\mu_1{=}1$, $\mu_2{=}1$, $p_2{=}0$. The crosses represent simulation points while solid curves are plotted from analytical results.}\label{ffig1}
\end{figure}
\begin{figure} [htbp]
   \centering      
        \begin{subfigure}[b]{0.46\textwidth}
                \centering
                \begin{tikzpicture}
                 \begin{axis}[legend style={draw=none},
                                     font=\tiny,
                                       label style={font=\tiny},
      			           xlabel={$\lambda_1$},
			           xlabel style={yshift=+9pt},		
                                       ylabel= {$\beta_1$},
                                       ylabel style={yshift=-16pt},
                                       xmin=1, xmax=5, ymin=0.0, ymax=0.8,
                                      width=\linewidth,height=\linewidth,legend pos=  north west]
                \addplot[thick, smooth,red] table[x=lbd1, y=beta1_p1_0.00] {ffig2.dat};
\addlegendentry{$p_1{=}0$}
\addplot[thick, smooth,green] table[x=lbd1, y=beta1_p1_0.35] {ffig2.dat};
\addlegendentry{$p_1{=}0.35$}
\addplot[thick, smooth,blue] table[x=lbd1, y=beta1_p1_0.70] {ffig2.dat};
\addlegendentry{$p_1{=}0.7$}
\addplot[thick, smooth,pink] table[x=lbd1, y=beta1_p1_1.00] {ffig2.dat};
\addlegendentry{$p_1{=}1$}
 \addplot[only marks,black,mark=x] table[x=lbd1, y=simu_beta1_p1_0.00] {ffig2_simu.dat};
\addplot[only marks,black,mark=x] table[x=lbd1, y=simu_beta1_p1_0.35] {ffig2_simu.dat};
\addplot[only marks,black,mark=x] table[x=lbd1, y=simu_beta1_p1_0.70] {ffig2_simu.dat};
\addplot[only marks,black,mark=x] table[x=lbd1, y=simu_beta1_p1_1.00] {ffig2_simu.dat};
                \end{axis}
      \end{tikzpicture}
       \hspace{0.05\linewidth}
       \caption{ Loss probability of class $1$}
       \label{ffig2_a}
        \end{subfigure}
   \quad
  \begin{subfigure}[b]{0.46\textwidth}
                \centering
                \begin{tikzpicture}
                 \begin{axis}[legend style={draw=none},
                                       font=\tiny,
                                       label style={font=\tiny},
      			           xlabel={$\lambda_1$},
			           xlabel style={yshift=+9pt},		
                                       ylabel= {$\beta_2$},
                                       ylabel style={yshift=-16pt},
                                       xmin=1, xmax=5, ymin=0.00, ymax=0.4,
                                      width=\linewidth,height=\linewidth,legend pos= north west]
                \addplot[thick, smooth,red] table[x=lbd1, y=beta2_p1_0.00] {ffig2.dat};\addlegendentry{$p_1{=}0$}
\addplot[thick, smooth,green] table[x=lbd1, y=beta2_p1_0.35] {ffig2.dat};
\addlegendentry{$p_1{=}0.35$}
\addplot[thick, smooth,blue] table[x=lbd1, y=beta2_p1_0.70] {ffig2.dat};
\addlegendentry{$p_1{=}0.7$}
\addplot[thick, smooth,pink] table[x=lbd1, y=beta2_p1_1.00] {ffig2.dat};
\addlegendentry{$p_1{=}1$}
\addplot[only marks,black,mark=x] table[x=lbd1, y=simu_beta2_p1_0.00] {ffig2_simu.dat};
\addplot[only marks,black,mark=x] table[x=lbd1, y=simu_beta2_p1_0.35] {ffig2_simu.dat};
\addplot[only marks,black,mark=x] table[x=lbd1, y=simu_beta2_p1_0.70] {ffig2_simu.dat};
\addplot[only marks,black,mark=x] table[x=lbd1, y=simu_beta2_p1_1.00] {ffig2_simu.dat};
                \end{axis}
      \end{tikzpicture}
       \hspace{0.05\linewidth}
       \caption{Loss probability of class $2$}\label{ffig2_b}
       \end{subfigure}
\caption{Loss probabilities as a function of $\lambda_1$ with $\lambda_2{=}12$, $c_1{=}1$, $c_2{=}10$,  $\mu_1{=}1$, $\mu_2{=}1$, $p_2{=}0$. The crosses represent simulation points while solid curves are plotted from analytical results.}\label{ffig2}
\end{figure}
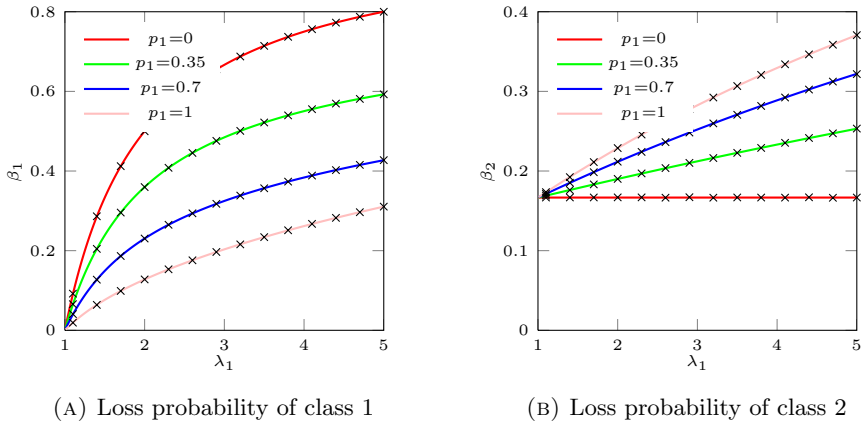
\begin{figure} [htbp]
   \centering      
        \begin{subfigure}[b]{0.46\textwidth}
                \centering
                \begin{tikzpicture}
                 \begin{axis}[legend style={draw=none},
                                  font=\tiny,
                                       label style={font=\tiny},
      			           xlabel={$p_1$},
			           xlabel style={yshift=+9pt},		
                                       ylabel= {$\beta_1$},
                                       ylabel style={yshift=-16pt},
                                       xmin=0, xmax=1, ymin=0.0, ymax=0.18,
                                      width=\linewidth,height=\linewidth,legend pos= north east]
                \addplot[thick, smooth,red] table[x=p1, y=beta1_lbd2_9.90] {ffig3.dat};
\addlegendentry{$\lambda_2{=}9.9$}
\addplot[thick, smooth,green] table[x=p1, y=beta1_lbd2_11.00] {ffig3.dat};
\addlegendentry{$\lambda_2{=}11$}
                \end{axis}
      \end{tikzpicture}
       \hspace{0.05\linewidth}
       \caption{ Loss probability of class $1$}
       \label{ffig3_a}
        \end{subfigure}
   \quad
  \begin{subfigure}[b]{0.46\textwidth}
                \centering
                \begin{tikzpicture}
                 \begin{axis}[legend style={draw=none},
                                 font=\tiny,
                                       label style={font=\tiny},
      			           xlabel={$p_1$},
			           xlabel style={yshift=+9pt},		
                                       ylabel= {$\beta_2$},
                                       ylabel style={yshift=-16pt},
                                       xmin=0, xmax=1, ymin=0., ymax=0.18,
                                      width=\linewidth,height=\linewidth,legend pos= north west]
                \addplot[thick, smooth,green] table[x=p1, y=beta2_lbd2_9.90] {ffig3.dat};\addlegendentry{$\lambda_2{=}9.9$}
\addplot[thick, smooth,red] table[x=p1, y=beta2_lbd2_11.00] {ffig3.dat};
\addlegendentry{$\lambda_2{=}11$}
                \end{axis}
      \end{tikzpicture}
       \hspace{0.05\linewidth}
       \caption{Loss probability of class $2$}\label{ffig3_b}
       \end{subfigure}
\caption{Loss probabilities as a function of $p_1$ with $c_1{=}1$, $c_2{=}10$, $\lambda_1{=}1.2$,  $\mu_1{=}1$, $\mu_2{=}1$, $p_2{=}0$. }\label{ffig3}
\end{figure}
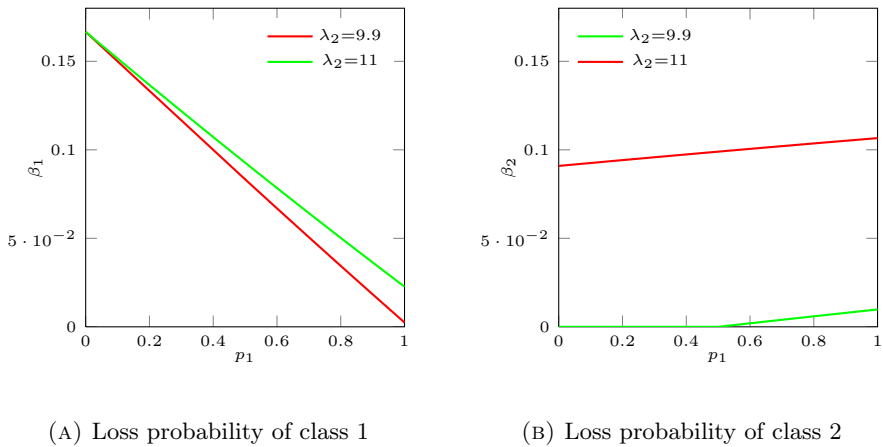

Figures~\ref{ffig3} illustrate the impact of the choice of $p_1$ when facility \#$2$ is almost overloaded, $\lambda_2=9.9$ so that $\lambda_2<c_2\mu_2$, and with a high load $\lambda_2=11.1$. As it can be seen, even when $p_1=1$, the performances of class $2$ requests are not really impacted by the offloading scheme, whereas the loss rate of class $1$ is significantly changed. This confirms the benefit of the offloading strategy. 

\section{Conclusion}
\label{conclusion}

We have proposed in this paper an analytical model to study a simple offloading strategy for data centers in the framework of fog computing under heavy loads. The strategy considered consists of forwarding with a certain probability requests blocked at a small data center to a big back-up data center. The model considered could also be used to study the offload of requests blocked at the big data center onto a small data center but this case has not been considered in the numerical applications. The key finding is that the proposed strategy can significantly improves blocking at a small data center without affecting too much blocking at the big data center. 

\providecommand{\bysame}{\leavevmode\hbox to3em{\hrulefill}\thinspace}
\providecommand{\MR}{\relax\ifhmode\unskip\space\fi MR }
\providecommand{\MRhref}[2]{%
  \href{http://www.ams.org/mathscinet-getitem?mr=#1}{#2}
}
\providecommand{\href}[2]{#2}

\end{document}